\documentclass[11pt]{article}

\usepackage{cite}
\usepackage{amsmath,amssymb,amsfonts}
\usepackage{algorithmic}
\usepackage{graphicx}
\usepackage{textcomp}
\usepackage{xcolor}
\def\BibTeX{{\rm B\kern-.05em{\sc i\kern-.025em b}\kern-.08em
    T\kern-.1667em\lower.7ex\hbox{E}\kern-.125emX}}

\usepackage[letterpaper,top=1in,bottom=1in,left=1in,right=1in]{geometry}
\usepackage{hyperref}
\usepackage{booktabs}
\newcommand{\ra}[1]{\renewcommand{\arraystretch}{#1}}
\usepackage{graphicx} 
\usepackage{enumerate}
\usepackage[shortlabels]{enumitem}
\usepackage{algorithm}
\usepackage{todonotes}
\usepackage{hyperref}
\usepackage{amsmath, amsfonts}

\usepackage{amsthm}
\usepackage[capitalize]{cleveref}
\usepackage{lineno}
\usepackage{thm-restate}
\usepackage{makecell}

\usepackage{xcolor}

\newtheorem{theorem}{Theorem}

\newtheorem{lemma}[theorem]{Lemma}
\newtheorem{definition}[theorem]{Definition}

\Crefname{@theorem}{Theorem}{Theorem}

\newcommand{\ol}{\overline}

\DeclareMathOperator{\height}{height}
\DeclareMathOperator{\rank}{rank}

\newcommand{\Sort}{\texttt{UniTopSort}}

  \usepackage{todonotes}
  \tikzset{notestyleraw/.append style={rectangle}}
\date{}

\begin{document}

\title{Sorting under Partial Information with Optimal Preprocessing Time via Unified Bound Heaps}

\author{Daniel Rutschmann \thanks{Institute of Science and Technology Austria (ISTA), Klosterneuburg, Austria.\\\url{https://orcid.org/0009-0005-6838-2628} }}

\maketitle


\maketitle
\begin{abstract}
In 1972, Fredman proposes the problem of sorting under partial information:
preprocess a directed acyclic graph $G$ with vertex set $X$ so that
you can sort $X$ in $O(\log e(G))$ time, where $e(G)$ is the number of sorted orders compatible with $G$.
Cardinal, Fiorini, Joret, Jungers and Munro [STOC'10] show that you can preprocess $G$ in $O(n^{2.5})$ time and then sort $X$ in $O(\log e(G) + n)$ time and $O(\log e(G))$ comparisons.
Recent work of van der Hoog and Rutschmann [FOCS'24] implies an algorithm with
$O(n^{\omega})$ preprocessing time where $\omega < 2.372$ and $O(\log e(G))$ sorting time.
Haeupler, Hladík, Iacono, Rozhoň, Tarjan and Tětek [SODA'25] achieve an overall running time of $O(\log e(G) + m)$.
In this paper, we achieve tight bounds for this problem: $O(m)$ preprocessing time and $O(\log e(G))$ sorting time.

As a key ingredient, we design a new fast heap data structure that might be of independent theoretical interest.
Our heap satisfies the \emph{unified bound}, which has previously been studied for dictionaries and binary search trees.

\end{abstract}

\paragraph{Funding.}
This research was conducted partly at the IT University in Copenhagen, Copenhagen, Denmark and partly at ISTA.
Daniel Rutschmann was supported by the Carlsberg Foundation Young Researcher Fellowship CF21-0302 -- ``Graph Algorithms with Geometric Applications''; and by the European Research Council (ERC) under the European Union's Horizon 2020 research and innovation programme (No.\ 101019564). \raisebox{-0.2cm}{\includegraphics[width=1.6cm]{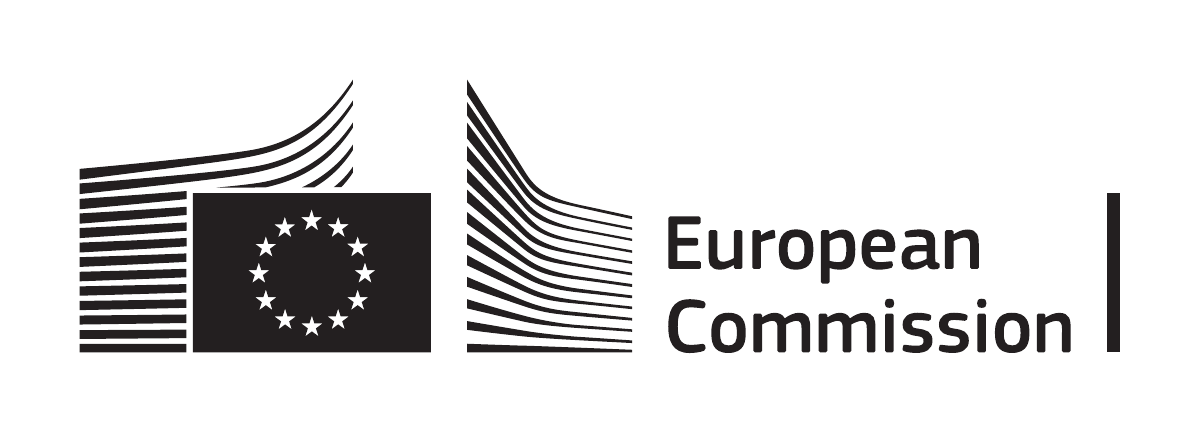}}

For open access purposes, the author has applied a CC BY public copyright license to this version of the manuscript.
Views and opinions expressed are those of the author(s) only and do not necessarily reflect those of the European Union or the European Research Council Executive Agency. Neither the European Union nor the granting authority can be held responsible for them.

\newpage

\section{Introduction}
Comparison based sorting is a fundamental problem in computer science.
You are given a list $X$ of $n$ items with some unknown linear order $L$.
You can obtain information about $L$ by comparing any pair of items,
and the goal is to output $X$ sorted along $L$.
Since there are $n!$ possible outputs, an information theoretic lower bound implies that
any algorithm needs to do $\Omega(n \log n)$ comparisons in the worst case.
Many textbook sorting algorithms (heapsort, mergesort, deterministic quicksort) match this bound.

In 1976, Fredman~\cite{fredman_how_1976} introduces a generalization of sorting, called \emph{sorting under partial information}:
In addition to the list $X$ with unknown linear order $L$,
you are given a directed acyclic graph (DAG) $G$ with vertex set $X$,
and the promise that every edge $x \to y$ of $G$ satisfies $x <_L y$.
Given this generalized input, the sorting under partial information problem asks:
Can you use this DAG to do fewer comparisons?
Let $m$ be the number of edges in $G$ and let $e(G)$ be the number of linear orders compatible with $G$.
An information theoretic lower bound shows that any algorithm needs to do $\Omega(\log e(G))$ comparisons in the worst case.
Fredman asks whether this bound is tight.
That is, given $G$, can you construct a decision
tree of height $O(\log e(G))$ that sorts $X$?

This setting is very general: If $G$ contains a Hamilton path, then $e(G) = 1$, so the decision tree should perform no comparisons
after analyzing $G$.
If $G$ consists of two disjoint paths of length $n/2$, then $\log e(G) \in \Theta(n)$ and this is the merge problem.
If $G$ contains no edges, then $\log e(G) \in \Theta(n \log n)$ and we are back at the vanilla sorting problem.
If $G$ is an arbitrary DAG, then $O(\log e(G))$ can be any value from $O(1)$ to $O(n \log n)$.

\paragraph{Previous work.} 
In 1976, Fredman~\cite{fredman_how_1976} introduces the problem of sorting under partial information.
Given $G$, he shows that you can construct a decision tree of
height $O(\log e(G) + n)$ that sorts $X$. This exceeds the $\Omega(\log e(G))$ lower bound by an additive $O(n)$ term.
In 1984, Kahn and Saks~\cite{kahn_balancing_1984} show that you can construct a decision tree of height $O(\log e(G))$, matching the information theoretic lower bound.
When implemented as a word RAM algorithm, both of their approaches need exponential time
to construct the decision tree.
This decision tree can then be simulated in $O(\log e(G) + n)$ or $O(\log e(G))$ time, respectively.

Cardinal, Fiorini, Joret, Jungers and Munro~\cite{cardinal_sorting_2013}
study this problem algorithmically.
They observe that the algorithmic equivalent of constructing a decision tree is a two-stage process:
In the \emph{preprocessing phase}, the algorithm can access $G$ but cannot do any comparisons yet.
Afterwards, there is a \emph{sorting phase} where the algorithm can do comparisons.
At the end of the sorting phase, the algorithm has to output $X$ sorted along $L$.
When the decision tree arguments by Fredman~\cite{fredman_how_1976}
and Kahn and Saks~\cite{kahn_balancing_1984} are converted to algorithms that match this framework,
these algorithms need exponential time and space in the preprocessing phase,
but achieve a fast sorting phase.
Kahn and Kim 1992~\cite{kahn_entropy_1992} design the first polynomial time algorithm that sorts $X$ in $O(\log e(G))$ comparisons.
Cardinal, Fiorini, Joret, Jungers and Munro~\cite{cardinal_sorting_2013} show that you can preprocess $G$ in $O(n^{2.5})$ time,
after which you can sort $X$ in $O(\log e(G))$ comparisons and $O(\log e(G) + n)$ time.
The sorting phase of their algorithm can be seen as a decision tree of height $O(\log e(G))$
that can be implemented by a word RAM algorithm using $O(\log e(G) + n)$ time and space.

For all the above algorithms, it doesn't really matter how we represent $G$ in the input.
There are two reasonable choices:
For graph algorithms, we typically assume a representation in $O(m)$ space,
such as a list of $m$ edges, or an adjacency list.
But mathematically speaking, the information encoded by $G$ is a partial order on the set $X$,
where we can query for any pair $x, y \in X$ whether $x \prec y$,
similar to an in-memory adjacency matrix.
The main difference between these representations
is that partial orders are transitively closed:
if $u \prec v$ and $v \prec w$, then also $u \prec w$.
Directed acyclic graphs need not have this property: there might be edges $u \to v$ and $v \to w$, but no edge $u \to w$.
We can compute the transitive closure of a directed acyclic graph in $O(n^{\omega})$ time,
where $\omega < 2.372$ is the exponent of fast matrix multiplication.
Since $\omega < 2.5$, this additional step does not
increase the running time of the above algorithms.
But for the algorithms below, this does make a difference.

If $G$ is given as a partial order in memory,
Van der Hoog and Rutschmann~\cite{HoogRutschmann24} achieve $O(n^{1 + \varepsilon})$ preprocessing time and $O(\varepsilon^{-1} \log e(G))$ sorting time, which is tight for this input model.
If $G$ is given as a DAG, then we first compute the transitive closure,
and then run their algorithm.
In the two phase algorithmic framework, this implies
an algorithm with $O(n^{\omega})$ preprocessing time
that then sorts $X$ in $O(\log e(G))$ comparisons and time.
If we insist on $O(\log e(G))$ sorting time,
then this is the fastest known preprocessing time.

Haeupler, Hladík, Iacono, Rozhoň, Tarjan and Tětek~\cite{haeupler2024fast}
and van der Hoog, Rotenberg and Rutschmann~\cite{van_der_hoog_simpler_2025}
study the model where $G$ is given as a DAG.
They design an algorithm with $O(m)$ preprocessing time and $O(\log e(G) + m)$ sorting time.
This is tight in terms of the \emph{overall running time}.
However, compared to previous algorithms, this does not achieve a sorting time of $O(\log e(G))$.
In this paper, we close this gap.
We design an algorithm that preprocesses $G$ in $O(m)$ time
and then sorts $X$ in $O(\log e(G))$ comparisons and time,
which is tight. See \cref{tabl:upper_bounds}.

\begin{table}
    \ra{1.2}
	\centering
    \begin{tabular}{@{}l@{}c@{}ll@{}c@{}l@{}}
        \toprule
        Preprocessing $G$ & \phantom{abc} & \multicolumn{2}{@{}l@{}}{Sorting $X$} & \phantom{abc} & Source \\
        \cmidrule{1-1} \cmidrule{3-4}
        Time & & Comparisons & Time && \\
        \midrule
        {\color{red} exponential} && {\color{red} $\log e(G) + O(n)$} & {\color{red} $O(\log e(G) + n)$} && \cite[TCS'76]{fredman_how_1976}\\
        {\color{red} exponential} && $O(\log e(G))$ & $O(\log e(G))$ && \cite[Order'84]{kahn_balancing_1984}\\
        $O(m)$ && $O(\log e(G))$ & {\color{red} poly($n$)} && \cite[STOC'92]{kahn_entropy_1992}\\
        {\color{red} $O(n^{2.5})$} && $O(\log e(G))$ & {\color{red} $O(\log e(G) + n)$} && \cite[STOC'10]{cardinal_sorting_2013}\\
        {\color{red} $O(n^{\omega})$} && $O(\log e(G))$ & $O(\log e(G))$ && \cite[FOCS'24]{HoogRutschmann24}\\
        $O(m)$ && $O(\log e(G))$ & {\color{red} $O(\log e(G) + m)$} && \cite[SODA'25]{haeupler2024fast}\\
        $O(m)$ && $O(\log e(G))$ & {\color{red} $O(\log e(G) + m)$} && \cite[SOSA'25]{van_der_hoog_simpler_2025}\\
        \midrule
        $O(m)$ && $O(\log e(G))$ & $O(\log e(G))$ && This work\\
        \bottomrule
    \end{tabular}
\caption{\label{tabl:upper_bounds} Previous algorithms for sorting under partial information, and our contribution. Non-optimal running times in red.}
\end{table}

\paragraph{Further related work}
In the geometric setting,
algorithms with separate preprocessing and sorting phases have been extensively studied under the
\emph{preprocessing model of uncertainty}.
Van der Hoog, Kostitsyna, L\"{o}ffler and Speckmann~\cite{van_der_hoog_preprocessing_2019}
show that you can preprocess a set of intervals in $O(n \log n)$ time, and then, given an item from each interval,
sort these items in $O(\log e(G))$ time, where $G$ is the DAG induced by the intervals.
Their result could be seen as the geometric analogue of sorting under partial information.
They generalize this to two dimensions~\cite{van_der_hoog_preprocessing_2022}.
Held and Mitchell~\cite{held_triangulating_2008} show that you can preprocess a set of disjoint disks in $O(n \log n)$ time, and then,
given a point from each disk compute some triangulation in linear time.
In the same setting, you can also compute the Gabriel graph, Euclidean MST, Delaunay triangulation (Löffler and Snoeyink~\cite{loffler_delaunay_2010}, Devillers~\cite{devillers_delaunay_2011}), or quad tree (Löffler and Mulzer~\cite{loffler_triangulating_2012}) in linear time after preprocessing;
or a $k$-layer onion decomposition (Löffler and Mulzer~\cite{loffler_unions_2013}) in $O(n \log k)$ time after preprocessing.
There are further generalizations to overlapping regions~\cite{buchin_preprocessing_2011},
regions of arbitrary shape~\cite{van_kreveld_preprocessing_2010}, or points on lines~\cite{ezra_convex_2013}.

\paragraph{Adaptive heap data structures}
A heap is a data structure that maintains a set of $n$ comparable items under two operations:
\textsc{insert}$(x)$ adds a new item to the set,
and $x \gets$\textsc{extractMin}$()$ removes the smallest item $x$ from the set and returns this item.
Some heaps additionally support the operation \textsc{decreaseKey}$(x, y)$,
which replaces the item $x$ by an item $y$ with $y < x$.
We require that \textsc{insert} runs in $O(1)$ (amortized) time and then seek to minimize the running time of \textsc{extractMin}. 
We say a heap is adaptive if \textsc{extractMin} runs in $o(\log n)$ time
for specific sequences of \textsc{insert} and \textsc{extractMin} operations.
Adaptive data structures are useful since they allow us to perform $\Theta(n)$ operations
while still achieving a sorting time of $o(n \log n)$.

Iacono~\cite{iacono_improved_2000} shows that the pairing heap
has the \emph{working set bound}: It supports $x \gets $\textsc{extractMin}$()$ in
$O(\log k)$ time where $k$ is the number of \textsc{insert} operations performed while $x$ is in the heap.
By using this result, Haeupler, Hladík, Iacono, Rozhoň, Tarjan and Tětek~\cite{haeupler2024fast}
obtain their $O(\log e(G) + m)$ algorithm for sorting under partial information.
Recently, adaptive heap data structures have led to fast algorithms for other problems:
Haeupler, Hlad{\'\i}k, Rozho{\v{n}}, Tarjan and T{\v{e}}tek~\cite{Haeupler24}
design a heap data structure with the working set bound
that also supports the decrease-key operation in $O(1)$ amortized time.
They show that Dijkstra's algorithm with their heap is universally optimal.

Furthermore, Elmasry~\cite{elmasry_priority_2006} designs a heap
that supports $x \gets$\textsc{extractMin}$()$ in $O(\log k)$ time where $k$ is the number of items
that get inserted after $x$ and remain in the heap until $x$ gets extracted.
Elmasry, Farzan and Iacono~\cite{elmasry_priority_2012} design a heap
that supports $x \gets$\textsc{extractMin}$()$ in $O(\min(\log k, \log \ell))$ time where $k$ is defined the same way and $\ell$ is the number of items
that get inserted \emph{before} $x$ and remain in the heap until $x$ gets extracted.
Kozma and Saranurak~\cite{kozma_smooth_2018} introduce the \emph{smooth heap}, which is conjectured to be instance optimal,
but seems difficult to analyze.

\paragraph{Contribution.} We provide the first algorithm for sorting under partial information
that is tight in both preprocessing and sorting time.
To this end, we introduce the notion of the \emph{unified bound} for heaps.
Consider an initially empty heap under \textsc{insert} and \textsc{extractMin} operations.
We allow some items to remain in the heap at the end of the execution.
Let $X$ be the set of items that get inserted into the heap.
For an item $x \in X$, let $a(x) \in [n]$ denote the number of items that get inserted no later than $x$ gets inserted.
If $x$ does not get extracted, put $b(x) = \infty$,
otherwise, let $b(x) \in [n]$ denote the number of items that get extracted no later than $x$ gets extracted.
In other words, $x$ is the $a(x)$-th item to be inserted,
and the $b(x)$-th item to be extracted.
\begin{definition} \label{def:unified}
    A heap data structure has the \emph{unified bound}
    if it supports \textsc{insert}$(x)$ in amortized $O(1)$ time,
    the very first \textsc{extractMin} operation in amortized $O(1)$ time,
    and subsequent $x \gets$\textsc{extractMin}$()$ operations in amortized time
    $$
    O\Big(\min\limits_{\substack{y \in X\\b(y) < b(x)}} \big(1 + \log 
    |a(x) - a(y)| + \log (b(x) - b(y))\big)\Big).
    $$
\end{definition}
Intuitively, the unified bound allows us to use items $y$ that were already extracted as a ``finger'' to
find the item $x$ that gets extracted next. The cost of using a particular item $y$ contains two terms: $\log |a(x) - a(y)|$, and $\log (b(x) - b(y))$.
Roughly speaking, $\log |a(x) - a(y)|$ measures how ``close'' $x$ and $y$ are in the insertion order, while $\log (b(x) - b(y))$ measures how ``old'' the finger $y$ is.
Crucially, we do not need to specify $y$ ourselves;
the data structure automatically selects an (asymptotically) optimal $y$.
Designing a heap with the unified bound is our main technical contribution:

\begin{theorem} \label{theo:heap_exists}
    There is a heap data structure with the unified bound.
\end{theorem}

We believe that this data structure and our notion of heaps with the unified bound
are of independent theoretical interest. For example, we show:

\begin{lemma}[Informal \cref{lemm:working_set_precise}] \label{lemm:unified_working_set}
    If a heap has the unified bound,
    then it also has the working set property.
\end{lemma}

Next, we use this heap to design a fast algorithm for sorting under partial information.
Let \Sort{} be the following two-phase algorithm:
In the preprocessing phase, \Sort{} computes a topological ordering of $G$.
This is the only preprocessing step, and takes $O(n + m)$ time.
In the sorting phase, \Sort{} considers all items of $X$ in topological order
and inserts them into a heap with the unified bound.
Then, it repeatedly extracts the minimum from the heap. This yields $X$ in sorted order.
Note that, to perform the \textsc{insert} and \textsc{extractMin} operations,
the heap has to compare items. But this is fine, since we have access to the linear order $L$ in the sorting phase.
\emph{A priori}, it seems unlikely that \Sort{} achieves $O(\log e(G) + n)$ sorting time.
Indeed, the only information about the graph $G$ that \Sort{} retains is a topological ordering.
Surprisingly, we show:
\begin{restatable}{theorem}{thmSorting} \label{theo:sorting_raw}
    \Sort{} has $O(n+m)$ preprocessing time and $O(\log e(G) + n)$ sorting time.
\end{restatable}

From \cref{theo:sorting_raw}, we apply existing reductions to reduce the preprocessing time to $O(m)$ and the sorting time to $O(\log e(G))$. The sorting time bounds the number of comparisons. This yields:

\begin{theorem} \label{theo:sorting_fast}
    We can preprocess a DAG $G$ on $X$ in $O(m)$ time, so that,
    given an unknown linear order $L$ on $X$
    and the guarantee that all edges in $G$ go from smaller items to larger items with respect to $L$,
    we can sort $X$ in $O(\log e(G))$ time and comparisons.
\end{theorem}

\paragraph{Comparison to existing algorithms for DAG sorting}
Both \Sort{} and the algorithm of Haeupler, Hladík, Iacono, Rozhoň, Tarjan and Tětek~\cite{haeupler2024fast}
are variants of Heapsort that use an adaptive heap data structure guided by a topological ordering of $G$.
Thus, at a surface level, these algorithms might look similar.
However, the two algorithms operate very differently. Let us illustrate three key differences:

(1) The algorithm of \cite{haeupler2024fast} tries to keep the heap as small as possible,
by inserting items on the fly once they become relevant (i.e. become sources in $G$).
Thus, the algorithm interleaves \textsc{insert} and \textsc{extractMin} operations.
\Sort{} inserts all $n$ items into the heap immediately.

(2) The algorithm of \cite{haeupler2024fast} is closely tied to the graph $G$:
Whenever an item is extracted from the heap, the algorithm looks at all outgoing edges from that item.
This yields the $O(m)$ term in the sorting time.
\Sort{} does not care about the individual edges of $G$ beyond using them to find a topological ordering of $G$.

(3) The heap used by \cite{haeupler2024fast} is much simpler than the one we use:
You can obtain a heap with the working set via a simple bucket-based construction
with heaps of exponentially increasing sizes~\cite{van_der_hoog_simpler_2025-1}.
Our construction of a heap with the unified bound is involved.
We have reasons to believe that an involved construction is necessary, as the
cost function in the unified bound contains two separate terms ($\log |a(x) - a(y)|$ and $\log (b(x) - b(y))$). We discuss prior attempts at this in \cref{sect:known_heaps}.

The other existing algorithms for DAG sorting do not use heaps:
\cite{van_der_hoog_simpler_2025} use a variant of insertion sort with a finger search tree,
where they look at all outgoing edges of an item to determine the optimal finger.
\cite{cardinal_sorting_2013} and \cite{HoogRutschmann24} decompose the graph into paths
and use a variant of merge sort based on Huffman trees.

\section{Preliminaries}
We denote $[n] = \{1, 2, \dots, n\}$ and we use $\log$ for the base-2 logarithm.
Let $X$ be the set of items and let $G$ be a DAG with vertex set $X$.
We say a linear order $L$ on $X$ is \emph{compatible with} $G$ if all edges in $G$ go
from a smaller item to a larger item under $L$, that is, $x <_L y$ for every edge $x \to y$.

\subsection{Existing Machinery for Sorting under Partial Information}

The interval lemma is a powerful tool for analysing fast algorithms that use adaptive data structures.
It was originally introduced in~\cite{cardinal_sorting_2013} and has been extensively used~\cite{cardinal_sorting_2013,Haeupler24,HoogRutschmann24,haeupler2024fast,van_der_hoog_simpler_2025}.
Here, we consider a lopsided version where the number of intervals may be different from
the range of the intervals.
For example, for the analysis in Section 3, it is convenient to consider $n$ intervals in $[1, 2n]$.

\begin{restatable}[A generalization of Lemma 3.2 in~\cite{cardinal_sorting_2013}]{lemma}{intervallemma} \label{lemm:interval}
    Let $\{ [a_i, b_i ] \}_{i=1}^{k}$ be a family of $k$ subintervals of $[0, m]$
    that each have size at least $1$.
    Let $Z$ be the set of linear orders realizable by real numbers $z_i \in [a_i, b_i]$.
    Then $ \sum\limits_{i \in [k]} \log( b_i - a_i) \in O(k \log (\frac{em}{k}) + \log |Z|)$.
\end{restatable}

\begin{proof}
    Many existing proofs of Lemma 3.2 in~\cite{cardinal_sorting_2013} can be adapted to work in the lopsided setting.
    For example, we tweak the proof of \cite{haeupler2024fast}:
    For each $i$ between $1$ and $k$ inclusive, choose a real number $z_i$ uniformly at random from the real interval $[0, m]$, independently for each $i$. With probability $1$, the $z_i$ are distinct. 
    Let $L$ be the permutation of $[k]$ obtained by sorting $[k]$ by $z_i$. Each possible permutation is equally likely.
    If each $z_i \in [a_i, b_i]$ then $L$ is in $Z$.
    The probability of this happening is $\prod_{i=1}^k \frac{b_i-a_i}{m}$.
    It follows that $|Z| \geq k! \cdot \prod_{i=1}^k \frac{b_i-a_i}{m}$.
    Taking logarithms gives $\log |Z| \geq \sum_{i=1}^k \log (b_i-a_i) + \log k! - k\log m$.  By Stirling's approximation of the factorial, $\log k! \geq k\log k - k\log e$. Thus,

    \vspace{-1em}
    \begin{align*}
        \sum_{i=1}^k \log (b_i-a_i) &\le \log |Z| + k \log m - \log k! \le \log |Z| + k \log \Big(\frac{e m}{k}\Big). \qedhere
    \end{align*}
\end{proof}

In the introduction, we present two algorithms for sorting under partial information with different running times:
\cref{theo:sorting_raw} and \cref{theo:sorting_fast}.
The main difference between these
is the additional $O(n)$ term in the sorting time of \cref{theo:sorting_raw}.
This term is only relevant if $\log e(G) \in o(n)$,
and there is existing machinery \cite{cardinal_sorting_2013,HoogRutschmann24,van_der_hoog_simpler_2025} for dealing with this case:
Roughly speaking, this can only happen if there is a very long path in $G$.
We can find this path during the preprocessing phase,
and not waste any comparisons on it during the sorting phase. Formally:

\begin{lemma}[Lemma 1 in \cite{van_der_hoog_simpler_2025} + Lemma 8 of~\cite{HoogRutschmann24}] \label{lemm:long_path}
    Let $P$ be a longest path in $G$. Then, $n - |P| \in O(\log e(G))$.
    Moreover, after performing $O(n+m)$ preprocessing on $G$, we can,
    given the sorted order of $X - P$, sort all of $X$ in $O(\log e(G))$ time during the sorting phase.
\end{lemma}

In \cref{sect:reduction} we use this lemma to derive \cref{theo:sorting_fast} from \cref{theo:sorting_raw}.

\subsection{Existing Results on Adaptive Data Structures}
\label{sect:known_heaps}

In this paper, we mainly focus on adaptive heap data structures.
However,
there is a standard way of using a binary search tree as a heap:
sort the items by the order in which they were inserted,
and maintain the subtree minimum at each node.
Thus, we will also mention some results on binary search trees,
as these results have implications for adaptive heaps.

\paragraph{Existing notions of unified bounds.}
Unified bounds were first defined for binary search trees.
For splay trees, Sleator and Tarjan~\cite{sleator_self-adjusting_1985} use the term
``unified bound'' for the minimum of three other bounds,
the so-called static optimality, static finger, and working set bound.
However, Elmasry, Farzan and Iacono~\cite{elmasry_priority_2012} show that this ``unified bound''
is equivalent to the working set bound. Hence, this definition has been obsoleted.

Iacono~\cite{iacono_alternatives_nodate} introduces the unified bound for binary search trees,
which is strictly stronger than the working set bound for binary search trees. 
He conjectures that the splay tree also satisfies this stronger bound.
This terminology is also used in~\cite{goos_simplified_2004,badoiu_unified_2007}. 
Elmasry, Farzan and Iacono~\cite{elmasry_priority_2012} attempt to define a unified
bound for Heaps,
based on Iacono's~\cite{iacono_alternatives_nodate} unified bound for binary search trees.
They use the term ``unified conjecture'' for this bound.
Their ``unified conjecture'' bound, using our notation from \cref{def:unified},
is $\min\limits_{y \in X:\ b(y) < b(x)} \big(1 + \log |\rank(x) - \rank(y)| + \log (b(x) - b(y)))$, where $\rank(x)$ is the number of items currently in the heap that are smaller than $x$.
Unfortunately, they show that, if a heap supports \textsc{insert} in amortized $O(1)$ time,
the running time of $x \gets$\textsc{extractMin}$()$ cannot be within the ``unified conjecture''.

\cref{def:unified} avoids this obstruction, since it uses the insert time $a(x)$ instead of the rank, and has an additional nice property: If you use a binary search tree
that has the unified bound (in the sense of Iacono~\cite{iacono_alternatives_nodate}) as a heap,
the resulting heap has the unified bound (in the sense of \cref{def:unified})%
\footnote{Curiously, from this viewpoint, the working set property for heaps~\cite{haeupler2024fast}
does not correspond to the working set property for binary search trees, but rather to a static finger search property for binary search trees
where the finger is the most recently inserted item.}.

\paragraph{Existing heaps fall short of the unified bound.}
None of the existing heap data structures~\cite{iacono_improved_2000,elmasry_priority_2006,elmasry_priority_2012} mentioned in the introduction are known to have the unified bound.
It is conceivable that the smooth heap or (some variant of) the pairing heap
has the unified bound --
we see this as an interesting but difficult open problem.

As discussed above, transforming a binary search tree with the unified bound into a heap
would yield a heap with the unified bound.
Unfortunately, there are no known binary search tree data structures with the unified bound.
The splay tree is conjectured to have the unified bound~\cite{iacono_alternatives_nodate},
but this conjecture remains open.
Resolving the dynamic finger conjecture for the splay tree
took considerable effort~\cite{cole_dynamic_2000}, and the unified bound would imply this conjecture. 
The binary search tree data structure of Derryberry and Sleator~\cite{dehne_skip-splay_2009} has an additive $O(\log \log n)$ overhead on top of the unified bound.
This would raise the running time of \Sort{} to $O(n \log \log n + \log e(G))$.
The dictionaries without this overhead~\cite{iacono_alternatives_nodate,goos_simplified_2004,badoiu_unified_2007} cannot be transformed into a heap, as they are not binary search trees.
Thus, there is no easy way of obtaining a heap with the unified bound.

\paragraph{Comparing the unified bound and the working set bound for heaps.}
For some sequences of \textsc{insert} and \textsc{extractMin} operations,
the unified bound is strictly stronger than the working set bound:
Consider inserting the items $1, n+1, 2, n+2, \dots, n, 2n$ into a heap,
and then repeatedly calling \textsc{extractMin} until the heap is empty.
A heap with the unified bound can perform these operations in $O(n)$ total time,
while the working set bound is only $O(n \log n)$.
Conversely, we now prove \cref{lemm:unified_working_set}: if a heap has the unified bound,
then it automatically has the working set property.

\begin{lemma} \label{lemm:working_set_precise}
    If a heap has the unified bound,
    then the \textsc{insert} operation runs in amortized constant time,
    and the $x \gets$\textsc{extractMin}$()$ operation
    runs in amortized $O(1 + \log(1+t - a(x)))$ time, where $t$
    is the number of \textsc{insert} operations performed so far
    and $a(x)$ is as in \cref{def:unified}.
\end{lemma}
Note that $t - a(x)$ is the number of \textsc{insert} operations that were performed while $x$ is in the heap, so $O(1 + \log(1+t - a(x)))$ is indeed the working set bound.
\begin{proof}
    Consider an arbitrary sequence $S$ of $t$ \textsc{insert} operations and $k$ \textsc{extractMin} operations.
    Let $x_1, \dots, x_k$ be the elements returned by the \textsc{extractMin} operations in order,
    and let $t_i$ be the value of $t$ at the time of the $x_i \gets$ \textsc{extractMin}$()$ operation.
    For all real numbers $\alpha, \beta > 0$, we have $\log(1+\alpha+\beta) \le \log(1+\alpha) + \log(1+\beta)$ and $\log(1+\alpha) \le \alpha$. Therefore,

    \begin{align*}
        \log(1 + |a(x_i) - a(x_{i-1})|)
        &\le \log(1+t_i-a(x_i)) + \log(1+t_{i-1} - a(x_{i-1})) + \log(1+t_i - t_{i-1})\\
        &\le \log(1+t_i-a(x_i))  + \log(1+t_{i-1} - a(x_{i-1})) + t_i - t_{i-1}
    \end{align*}
    In the unified bound (\cref{def:unified}), we will use $y = x_{i-1}$ for $x = x_i$.
    By the unified bound, the total cost to execute $S$ is

    \vspace{-1em}
    \begin{align*}
        &\ t + 1 + \sum_{i=2}^{k} \big(1 + \log(b(x_i) - b(x_{i-1})) + \log|a(x_i) - a(x_{i-1})| \big)\\
        \le&\  t + 1 + k + \sum_{i=2}^{k} (t_i - t_{i-1}) + \sum_{i=2}^{k} \big(\log(1+t_i - a(x_i)) + \log(1+t_{i-1} - a(x_{i-1})) \big) \\
        \le&\ t + 1 + k + 2 \sum_{i=1}^{k} \log(1+t_i - a(x_i)) + t,
    \end{align*}
    which is the total time to execute $S$ as claimed by \cref{lemm:working_set_precise}.
\end{proof}


\section{Analyzing \Sort}

In this section, we show:

\thmSorting*

The only preprocessing step in \Sort{} is computing a topological ordering of $G$.
This takes $O(n + m)$ time using a textbook algorithm.
In the sorting phase of \Sort{}, every item in $X$ gets inserted and extracted exactly once.
We assign every $x \in X$ two numbers $a(x), b(x) \in [1, n]$ as in \cref{def:unified}.
Then, we associate to every item $x$ the open interval $(\ell(x), r(x)) \subseteq [0, 2n]$,
where $r(x) = a(x) + b(x)$ and $\ell(x) = \max\Big(0, \max\limits_{y:\, a(y) < a(x) \land b(y) < b(x)} (a(y) + b(y) \big)\Big)$.

\begin{lemma} \label{lemm:pop_bound}
    The $x \gets$ \textsc{extractMin}$()$ operation has an amortized running time of $O(1 + \log(r(x) - \ell(x)))$.
\end{lemma}
\begin{proof}
    Intuitively, this follows from the fact that the definition of $\ell(x)$ takes the maximum over a smaller set
    compared to the unified bound.
    Formally, let $y_0 = \arg \max_{y:\, a(y) < a(x) \land b(y) < b(x)} (a(y) + b(y))$.
    Suppose first that $y_0$ does not exist (i.e. because there are no valid items $y$). Then $\ell(x) = 0$.
    If $x$ is the very first item extracted, then by the unified bound, the \textsc{extractMin} operation takes $O(1)$ time.
    Otherwise, let $y$ be an arbitrary item extracted before $x$. By the unified bound,
    the \textsc{extractMin} operation that returns $x$ takes $O(\log |a(x) - a(y)| + \log |b(x) - b(y)|)$
    time. Since

    \vspace{-1em}
    \begin{align*}
        \log |a(x) - a(y)| + \log|b(x) - b(y)| \le \log a(x) + \log b(x) \le 2 \log(a(x) + b(x)) = 2\log (r(x) - \ell(x)).
    \end{align*}
    This shows the lemma.
    Suppose now that $y_0$ exists. Then $\ell(x) = a(y_0) + b(y_0)$, and also $a(y_0) < a(x)$ and $b(y_0) < b(x)$.
    In the unified bound (\cref{def:unified}),
    we could pick $y = y_0$, hence the \textsc{extractMin} operation that returns $x$ takes
    $O(\log (a(x) - a(y)) + \log (b(x) - b(y)))$ amortized time. Since

    \vspace{-1em}
    \begin{align*}
        \log |a(x) - a(y)| + \log |b(x) - b(y)| \le 2 \log (a(x) + b(x) - a(y) - b(y)) =  2 \log (r(x) - \ell(x)),
    \end{align*}
    this shows the lemma.
\end{proof}

To bound the total running time of all \textsc{extractMin} operations,
we will eventually apply the interval lemma (\cref{lemm:interval}) to \cref{lemm:pop_bound}.
To this end, we first show:
\begin{lemma} \label{lemm:realization}
    For every $x \in X$, pick a distinct real number $w(x) \in [\ell(x), r(x)]$.
    Then, the sorted order $<_w$ on $X$ induced by $w$ is compatible with $G$.
\end{lemma}
\begin{proof}
    Let $y \to x$ be an edge of $G$. In any topological order of $G$, the item $y$ appears before $x$.
    Since \Sort{} inserts the items into the heap in topological order, this shows that $a(y) < a(x)$.
    By assumption, the unknown linear order $L$ on $X$ is compatible with $G$.
    Therefore, $y <_L x$, so $y$ gets extracted from the heap before $x$ does.
    This shows that $b(y) < b(x)$.
    By the definition of $\ell(x)$, this implies $\ell(x) \ge a(y) + b(y)$.
    Hence, $w(y) \le r(y) = a(y) + b(y) \le \ell(x) \le w(x)$.
    Since the values $w(.)$ are distinct, $w(y) < w(x)$, i.e. $y <_w x$. This shows the lemma.
\end{proof}

We now conclude the proof of \cref{theo:sorting_raw}:
In the sorting phase, \Sort{} first inserts all items into the heap,
and then extracts all items from the heap. The $n$ \textsc{insert} operations take $O(n)$ amortized time.
By \cref{lemm:pop_bound}, the running time of all $n$ \textsc{extractMin} operations is $O(n + \sum_{x \in X} \log (r(x) - \ell(x)))$.
By \cref{lemm:interval}, this sum is $O(n + \log |Z|)$,
and by \cref{lemm:realization}, $|Z| \le e(G)$.
Thus, the sorting phase of \Sort{} takes $O(n + \log e(G))$ time,
and hence performs $O(n + \log e(G))$ comparisons.
This shows \cref{theo:sorting_raw}.

\section{Deriving \cref{theo:sorting_fast} from \cref{theo:sorting_raw}}
\label{sect:reduction}

In this section, we use \cref{lemm:long_path} to derive a black-box reduction
that improves the time of the sorting phase from $O(\log e(G) + n)$ to $O(\log e(G))$:

\begin{theorem} \label{lemm:reduction}
    Given $G$, in $O(n+m)$ preprocessing time we can compute
    a subset $Y \subseteq X$ with $|Y| \in O(\log e(G))$
    and a graph $H$ on $Y$ with $O(m)$ edges,
    such that if $L$ is a linear order on $X$ compatible with $G$,
    then the restriction of $L$ to $Y$ is compatible with $H$.
    Moreover, every linear order on $Y$ compatible with $H$ can be obtained in this way,
    thus $e(H) \le e(G)$.
    Then, in the sorting phase, given the sorted order of $Y$,
    we can compute the sorted order of $X$ in $O(\log e(G))$ time.
\end{theorem}
\begin{proof}
    We first compute a longest path $P$ in $G$.
    The high-level idea is that we can shortcut $P$ down to $O(n - |P|)$ vertices
    without losing any information.
    Formally, let the predecessor (or successor) in $P$ of a vertex $u \in X - P$
    be the last (or first) vertex in $P$ that has an edge to (from) $u$.
    We define $Y$ to be $X - P$, plus the predecessor and successor in $P$ of each vertex in $X - P$.
    Therefore, $|Y| \le 3\, |X - P| \in O(\log e(G))$ by \cref{lemm:long_path}.
    We define $H$ to be the graph induced by $G$ on $Y$, plus additional edges to make $Y \cap P$
    a path with the same order as $P$.
    By construction, $H$ contains $O(m)$ edges, and can be computed in $O(n+m)$ time.

    Now, if $L$ is a linear order compatible with $G$, then $L$ in particular is compatible with $P$.
    Since every edge in $H$ is either an edge of $G$, or an edge on $P$ in the right direction,
    the restriction of $L$ to $Y$ is compatible with $H$.
    Conversely, given a linear order $L'$ on $Y$ compatible with $H$,
    we create a linear order $L$ on $X$ by placing each item in $P - Y$ into $L'$ right after its predecessor in $P$.
    This ensures that $L$ is compatible with all edges of $G$ that are incident to two vertices of $P$.
    $L$ is also compatible with all edges of $G$ incident to two vertices not on $P$, since these edges are also in $H$.
    Finally, for any edge $x \to y$ in $G$ between a vertex on $P$ and a vertex not in $P$, there is a
    path in $G$ that uses only edges on $P$, or from a predecessor in $P$ or to a successor in $P$.
    Thus, $L$ is also compatible with these edges, so $L$ is compatible with $G$.
    In particular, this shows $e(H) \le e(G)$.

    During preprocessing, we compute $H$ and perform the preprocessing of \cref{lemm:long_path}.
    All of this takes $O(n+m)$ preprocessing time.
    In the sorting phase, if we are given the sorted order of $Y$, then we also have the sorted order of $X - P$ since $X - P \subseteq Y$.
    Thus, by \cref{lemm:long_path} we can sort $X$ in $O(\log e(G))$ time.
\end{proof}

Using \cref{lemm:reduction}, we show that \cref{theo:sorting_raw} implies \cref{theo:sorting_fast}.
\begin{proof}
    We first check whether $m < n/3$. If so,
    then we ignore $G$ and sort $X$ from scratch with merge sort.
    Thus, we may now assume that $m \in \Omega(n)$.
    Given the DAG $G$ on $X$, we apply \cref{lemm:reduction}
    to obtain a DAG $H$ on a subset $Y \subseteq X$.
    We run \Sort{}$(Y, H)$ to obtain the sorted order on $Y$.
    Finally, we use the sorting phase of \cref{lemm:reduction} to get the sorted order on $X$.

    We first check that this is a correct algorithm:
    Let $L$ be a linear order compatible with $G$.
    By \cref{lemm:reduction}, the restriction of $L$ to $Y$ is compatible with $H$.
    Thus, $(Y, H)$ is a valid input for \Sort{} and we sort $Y$ correctly.
    Then, \cref{lemm:reduction} computes the sorted order of $X$.

    Next, we analyze the running time:
    If $m < n/3$, there are at least $n/3$ items not incident to any edges in $G$,
    so $\log e(G) \in \Theta(n \log n)$. We spend $O(m)$ preprocessing time
    and $O(n \log n) = O(\log e(G))$ sorting time in this case.
    Assume now that $m \in \Omega(n)$ and let $n', m'$ denote the number
    of vertices and edges of $H$.
    During preprocessing,
    \cref{lemm:reduction} spends $O(n+m) = O(m)$ time
    and \Sort{} spends $O(n' + m') = O(m)$ time.
    During sorting, \Sort{} spends $O(n' + \log e(H)) = O(\log e(G))$ time,
    and \cref{lemm:reduction} spends $O(\log e(G))$ time.
    Thus, we obtain an algorithm with $O(m)$ preprocessing and $O(\log e(G))$ sorting time.
    This shows \cref{theo:sorting_fast}.
\end{proof}

\section{Heaps with the Unified Bound} \label{sect:heap_construction}

In this section, we show \cref{theo:heap_exists} and design a heap with the unified bound (\cref{def:unified}). We achieve this by maintaining a tree where each leaf stores an item,
every node stores its subtree minimum (a pointer to the smallest item in its subtree),
and where the depth of each leaf is bounded by a function that looks like the unified bound.
To this end, consider $n$ leaves numbered $1, \dots, n$
and let $t$ be a time counter that starts at $0$.
Every leaf $i$ has a last access time $s(i) < t$.
We define $d(i) := 1 + \min_{1 \le j \le n} (\log (1+|i-j|) + \log (1 + t - s(j)))$.
In this section, we show:

\begin{restatable}{theorem}{thmTree} \label{theo:tree_heap}
    There is a data structure that maintains a rooted binary tree
    on top of $n$ leaves numbered $1, \dots, n$ that each store an item,
    such that each node in the tree stores its subtree minimum
    and every leaf $i$ has depth $O(d(i))$.
    The data structure supports the operations:
    \begin{itemize}
        \item \textsc{addLeaf}$(x)$: In amortized $O(1)$ time, increment $n$ and add a new leaf with index $n$, value $x$, and $s(n) = -\infty$.
        \item \textsc{access}$(i)$: In amortized $O(d(i))$ time (computed before the operation), set $s(i) = t$ and then increment $t$. For the very first access operation, $d(i)$ is ill-defined, so we require a bound of amortized $O(1)$ time instead.
        \item \textsc{changeKey}$(i, x)$: In $O(d(i))$ time, replace the item stored at leaf $i$ with $x$.
    \end{itemize}
    The data structure may remodel the internal tree nodes arbitrarily (within the allowed time bounds).
    The leaves do \emph{not} need to be stored in order.
\end{restatable}

But first, we show how \cref{theo:tree_heap} yields a heap with the unified bound (in the sense of \cref{def:unified}):
Suppose that \cref{theo:tree_heap} is true and let $T$ be the data structure that maintains a rooted tree.
For every \textsc{insert}$(x)$ operation, we call T.\textsc{addLeaf}$(x)$.
For every \textsc{extractMin}$()$ operation, we look at the subtree minimum of the root of $T$
to determine the leaf $i$ that contains the current minimum item $x$.
We call T.\textsc{access}$(i)$ and then T.\textsc{changeKey}$(i, \infty)$ where $\infty$ is a dummy item with value $\infty$.
Finally, we return $x$.
Observe that, at any point in time, the non-dummy items at the leaves
are precisely the items that have not been extracted.
Thus, the subtree minimum of the root of $T$
is the minimum of all these items, and this is a correct implementation of \textsc{insert} and \textsc{extractMin}.

As for the running time, the \textsc{insert} operation only calls
\textsc{addLeaf}, which takes amortized $O(1)$ time.
To analyze the \textsc{extractMin} operation, we first observe
that an item $x$ always gets stored in the leaf $i = a(x)$.
Moreover, the only calls to \textsc{access} we perform
are a single call in every \textsc{extractMin} operation.
Thus, after extracting the item $x$, we have $s(i) = b(x)$.
Consider now a call to \textsc{extractMin} that returns an item $x$ stored at a leaf $i$.
Recall that the unified bound is
$U := 1 + \min_{y:\, b(y) < b(x)} \big(\log |a(x) - a(y)| + \log (b(x) - b(y))\big)$.
Let $y_0$ be the value of $y$ that achieves the minimum and let $j_0$
be the index of the leaf that used to store $y_0$ before it was extracted.
Since $b(y_0) < b(x)$, the item $y_0$ was extracted before $x$. Thus, at this moment,
we have $s(j_0) = b(y_0)$. This implies

\vspace{-1em}
\begin{align*}
     d(i) &\le 1+ \log(1+|i-j_0|) + \log(1+t-s(j_0))\\
     &= 1+\log(1+|a(x) - a(y_0)|) + \log(1+b(x) - b(y_0)) \in O(U).
\end{align*}
Hence, the amortized running time of \textsc{access} during the \textsc{extractMin} operation
is indeed bounded by the unified bound.
After running \textsc{access}$(i)$, we have $d(i) \in O(1)$,
so the amortized running time of \textsc{changeKey}$(i, \infty)$ is $O(1)$.
This shows that \cref{theo:tree_heap} implies \cref{theo:heap_exists}.
In the remainder of this section, we will prove \cref{theo:tree_heap}.

\subsection{Overview}

At a high level, our data structure consists of two parts.
First, the \emph{main tree}: a complete binary tree with $\Theta(n)$ leaves,
ordered by index.
Some nodes in this tree are \emph{leaders},
where every leaf $i$ lies in the subtree of exactly one leader.
We call this leader the \emph{leader of leaf $i$}.
Second, the \emph{index trees}: we store the leaders in $O(\log \log n)$ balanced binary search trees of doubly exponential sizes.
Each index tree is ordered by the left-to-right order of the leaders in the main tree.
From the main tree and the index trees, we obtain a single tree, as required by \cref{theo:tree_heap}, as follows:
connect the roots of the index trees from big to small,
and attach to each leader node in the index tree its subtree in the main tree.
This yields some nodes with three children, but we can subdivide these nodes to get a binary tree.
(This increases the depth of each leaf by at most a factor $2$.)

In \cref{theo:tree_heap}, leaf $i$ should have depth $O(d(i))$.
Thus, in our construction,
for any leaf $i$ with leader $v$, both the height of the subtree of $v$ in the main tree,
and the height of the index tree that contains $v$ should be bounded by $O(d(i))$.
We will implement the \textsc{access}$(i)$ operation by moving the leaders down in the main tree,
with the goal that leaf $i$ ends up at depth $O(1)$.
Our carefully chosen invariants then guarantee that nearby leaves $j$ also have depth $O(d(j))$.
We implement an additional operation \textsc{cleanup} that moves the leaders up in the main tree,
with the goal of reducing the size of the index trees.
We invoke \textsc{cleanup}$()$ after every \textsc{addLeaf} or \textsc{access} operation.
Finally, it is easy to implement \textsc{changeKey}: since the values of items have no effect on the shape of our structure,
we just update all subtree minima on the root-to-leaf path.

The main difficulty is maintaining the subtree minima in the index trees as we create and destroy leaders:
while there are balanced binary trees with $O(1)$ insertion or deletion after lookup,
updating the subtree minima will take time linear in the height of the index tree.
To deal with this, we use a charging argument
where we insert $\log S$ consecutive leaders into an index tree at once, where $S$ is the size of the index tree.
Then, updating the subtree minima only takes $O(1)$ time per leader inserted.

\subsection{Invariants}

We first describe the \emph{main tree} $M$:
This is a complete binary tree with $m \in [n+1, 2n]$ leaves, numbered $1, 2, \dots, m$ where $m$ is some power of two.
The first $n$ of these leaves store items as in \cref{theo:tree_heap},
while the remaining leaves are there for future \textsc{addLeaf} operations.
Recall the definition of $t, s(i)$ and $d(i)$.
We say the \emph{height} of a node is the height of its subtree in $M$.
(Leaves have height $0$.)
Our goal is to maintain a set of leader nodes so that, for every $i$,
the subtree of the leader of leaf $i$ has height $O(d(i))$.
To this end, every node in $M$ can be \emph{active} or \emph{passive}.
We say a node $u$ is a \emph{leader} if $u$ is active and the parent of $u$ is passive.
We require the following invariants:
\begin{enumerate}[(1)]
    \item \label{inv:1} The leaves $1, \dots, n$ are active, and the leaves $n+1, \dots m$ are passive.
    \item \label{inv:2} If a node is active, then its children are active.
    \item \label{inv:3} Consecutive leaders have a height difference of at most two. Formally: Let $u_1, u_2, \dots$ be the leaders in left to right order. Then $|\height(u_j) - \height(u_{j+1})| \le 2$ for every $j$.
    \item \label{inv:4} For every leaf $i$, the $\lceil\log_2(1+t-s(i))\rceil$-th ancestor of $i$ is passive.
    \item \label{inv:5} The parent of leaf $n$ is passive.
\end{enumerate}
We will soon add some additional invariants to guarantee that the number of leaders stays small.
Indeed, with just the invariants \ref{inv:1}-\ref{inv:5}, we could make every leaf its own leader.
However, we first show that these invariants already imply that the depth of each leaf is within the unified bound.
\begin{lemma} \label{lemm:unique_leader}
    Every active node (and each leaf $1, \dots, n$) lies in the subtree of a unique leader.
\end{lemma}
\begin{proof}
    We first show that the root of $M$ is passive.
    Indeed, since $m \ge (n+1)$, the leaf $n+1$ exists.
    By invariant \ref{inv:1}, this leaf is passive, so by invariant \ref{inv:2}, the root of $M$ is passive too.

    Let $u$ be an arbitrary active node. In particular, by invariant \ref{inv:1}, $u$ could be any leaf of index $1, \dots, n$.
    Let $v$ be the first active node on a root-to-$u$ path.
    By invariant \ref{inv:2}, all nodes after $v$ are active.
    We have seen that $v$ cannot be the root of $M$.
    Thus, by definition, $v$ is the only leader on the root-to-$u$ path.
\end{proof}

Invariants~\ref{inv:3} and~\ref{inv:4} imply that the depth of each leaf is bounded by the unified bound:

\begin{lemma} \label{lemm:leader_unified}
    Let $v$ be the leader of the $i$-th leaf.
    Then $\height(v) \le d(i) + 3$.
\end{lemma}
\begin{proof}
    Let $u, v, w$ be consecutive leaders, that is,
    in left to right order, let $u$ be the leader before $v$,
    and let $w$ be the leader after $v$.
    By invariant \ref{inv:3}, $\height(u) \ge \height(v) - 2$ and $\height(w) \ge \height(v) - 2$.
    In particular, if $|j-i| \le 2^{\height(v) - 2}$, then the $j$-th leaf lies in
    one of the subtrees of $u, v, w$, so by invariant \ref{inv:4}, we have $1+t-s(j) \ge 2^{\height(v) - 3}$.
    For all other $j$, we have $|j - i| \ge 2^{\height(v) - 2}$.
    Therefore, $d(i) = 1 + \min_{1 \le j \le n} (\log (1+|i-j|) + \log (1 + t - s(j))) \ge \height(v) - 3$.
\end{proof}
The astute reader might notice that we have not used invariant \ref{inv:5} so far.
Indeed, \ref{inv:5} will only be used to support \textsc{addLeaf} in $O(1)$ time,
where it allows us to activate leaf $n+1$ without violating invariants \ref{inv:2} or \ref{inv:4}.

We now describe the index trees $L_0, L_1, \dots$. Each index tree $L_i$ is a balanced binary search tree.
Every leader is stored in exactly one index tree,
and leaders in the same index tree are ordered by their left-to-right order in the main tree $M$.
These index trees have additional invariants:
\begin{enumerate}[(1),start=6]
    \item \label{inv:b}
        Before every \textsc{addLeaf} or \textsc{access} operation, the size of $L_i$ is at most $20 \cdot 4^{2^{i+2}}$.
        After an \textsc{addLeaf} or \textsc{access} operation, the size of $L_i$ is at most $23 \cdot 4^{2^{i+2}}$.
    \item \label{inv:c} All leaders in $L_0$ have a height in $[0, 4]$, and for $i \ge 1$, all leaders in $L_i$ have a height in $[2^i, 4 \cdot 2^i]$.
        (Recall that the height of a leader is its height in $M$, not in $L_i$.)
    \item \label{inv:d} The leaders in $L_i$ of height $(2 \cdot 2^i, 4 \cdot 2^i]$ in $L_i$ are called \emph{unpromoted}.
        Among every maximal sequence $u_\ell, \dots, u_r$ of consecutive leaders in $M$ that all are in $L_i$ and are all unpromoted,
        we require there be at least one leader of height $2 \cdot 2^i+1$ or $2 \cdot 2^{i} + 2$. We call such a sequence $u_\ell, \dots, u_r$ an \emph{unpromoted block}.
\end{enumerate}
Invariant \ref{inv:d} might seem artificial at first, but it will allow us to move many leaders from $L_i$ to $L_{i+1}$ at once.
Indeed, whenever there is a leader of height $4 \cdot 2^i$ in $L_i$,
the unpromoted block in invariant \ref{inv:d} contains $\Omega(2^i)$ consecutive unpromoted leaders in $M$.
By promoting a whole block from $L_i$ to $L_{i+1}$, we will save a bunch of time updating subtree minima.
Invariants \ref{inv:b} and \ref{inv:c} allow us to bound the height of the index trees:

\begin{lemma} \label{lemm:leader_height}
    Let $u$ be a leader that is stored in some index tree $L_i$.
    Then $\height(L_i) \in O(\height(u))$.
\end{lemma}
\begin{proof}
    By Invariant \ref{inv:c}, we have $\height(u) \ge 2^i$ or $i = 0$.
    By Invariant \ref{inv:b} and (7), we have $|L_i| \le 23 \cdot 4^{2^{i+2}}$.
    Thus, $\height(L_i) \le 2^{i+2} \log(4) + \log(23) \in O(2^i)$.
\end{proof}

\subsection{Storing the data structure} We use a ``level-linked'' representation of the main tree:
every node stores pointers to its parent and its children, plus pointers to the left and right neighbor of the same height.
Every internal node stores a Boolean indicating whether it is active or not.
Every leaf $i$ stores its index $i$, its item, and its last access time $s(i)$.
Each active internal node moreover stores the maximum access time in its subtree,
and a pointer to the minimum item in its subtree.
In combination with invariant~\ref{inv:3},
this level-linked representation allows us to go from a leader to the next or previous leader by following $O(1)$ pointers
and checking whether $O(1)$ nodes are active or not.

For the index trees, we use any balanced binary tree that allows us to
insert or delete $k$ (presorted) consecutive items in $O(k + \log S)$ time where $S$ is the size of the balanced binary tree.
We call this a \emph{bulk insertion / deletion}.
If a balanced binary tree $T$ of size $S$ can be built in $O(S)$ time and supports split and join operations in $O(\log S)$,
then we can do a bulk insertion by splitting $T$ into $T_1, T_2$ at the insertion location,
building a tree $T'$ on the $k$ presorted elements, and joining $T_1, T', T_2$.
Similarly, we can perform a bulk deletion by splitting $T$ into $T_1, T', T_2$ and joining $T_1, T_2$.
One example of a balanced binary tree with these properties is a red-black tree~\cite{tarjan_data_1983}.

Every node in an index tree stores its subtree minimum, 
and the size of its subtree. We can update these values during a split or join operation in $O(\log S)$ time.
In particular, while maintaining these values, we can still do bulk insertions and bulk deletion in $O(k + \log S)$ time.
For every leader, we store pointers between its node in the main tree
and its node in the index trees.

\subsection{Amortized Analysis Potential}
We use the accounting method~\cite{cormen01introduction}:
creating a coin or bill takes $O(1)$ amortized time,
you can spend a coin to do $O(1)$ work for free, and spend a bill to gain $O(1)$ coins.
We have the following money:
\begin{itemize}[noitemsep]
    \item two coins per leader node,
    \item one additional coin per unpromoted leader node, and
    \item one bill per pair of consecutive leaders with height difference $2$.
\end{itemize}

We think it is easier to keep track of money locally.
But, for readers that prefer working with potential functions, you may think of the potential function 

\vspace{-1em}
\[
    \Phi = (\text{total number of coins}) + 10^9 \cdot (\text{total number of bills}).
\]

\subsection{The \textsc{access} operation}
Recall that a leader is an active node in the main tree with an inactive parent.
We perform \textsc{access} in two steps.
First, we ignore the index trees and only update the main tree.
We deactivate a bunch of nodes, which changes the set of leaders.
Second, we bulk delete the destroyed leaders from the index trees,
and bulk insert the newly created leaders into the index trees.
This approach saves a bunch of time updating subtree minima in the index trees.

To implement \textsc{access}$(i)$, first deactivate all non-leaf nodes on the root-to-leaf path of leaf $i$.
Then,
while there are consecutive leaders $u, v$ with $\height(u) - \height(v) \ge 3$, deactivate
$u$ and the right child of $u$. Similarly, whenever $\height(v) - \height(u) \ge 3$, deactivate $v$ and the left child of $v$.
Finally, erase all destroyed leaders from the index trees,
and insert all newly created leaders into the index trees.
Inserting or deleting leaders one-by-one would be too slow.
Instead, for every $i = 0, 1, \dots$, we build a linked list $N_i$ on all newly created leaders
of height $(2 \cdot 2^i, 4 \cdot 2^i]$ in left-to-right order,
or height $[0, 4]$ for $i=0$.
Then, for each affected index tree $L_i$,
we bulk delete the destroyed leaders from $L_i$,
and bulk insert $N_i$ into $L_i$.
Note that, except for leaders of height $[0,2]$ in $L_0$, all newly created leaders are unpromoted (invariant \ref{inv:d}).

We now analyze this implementation of \textsc{access}$(i)$:

\begin{lemma} \label{lemm:contiguous}
    At any point in time during and after the \textsc{access} operation,
    the set of newly created leaders is \emph{contiguous}, that is,
    it forms a contiguous subsequence of all leaders in the left-to-right order.
\end{lemma}
\begin{proof}
    Let $w$ be the leader of leaf $i$ before the \textsc{access} operation.
    First, we deactivate all non-leaf nodes on the path from $w$ to $i$.
    The newly created leaders are precisely the subset of leaders that lie in the subtree of $w$.
    Hence, they are contiguous.
    After that, we always find two consecutive leaders $u, v$ with $|\height(u) - \height(v)| \ge 3$
    and deactivate one of them, say $u$. Note that, since invariant \ref{inv:3} held before the \textsc{access} operation,
    at least one of $u, v$ has to be a newly created leader.
    Deactivating $u$ replaces it with two newly created leaders,
    so the set of newly created leaders remains contiguous.
\end{proof}

\begin{lemma} \label{lemm:monotone}
    At any point in time during and after the \textsc{access}$(i)$ operation,
    the heights of newly created leaders to the left of leaf $i$ are non-increasing,
    and the heights of newly created leaders to the right of leaf $i$ are non-decreasing.
    Moreover, at most two consecutive newly created leaders have the same height.
\end{lemma}
\begin{proof}
    This is true after we deactivate all non-leaf nodes on the root-to-leaf path to node $i$.
    We show that \cref{lemm:monotone} is maintained as we deactivate more nodes.
    Let $t, u, v$ be consecutive leaders with $\height(u) - \height(v) \ge 3$.
    By assumption, \cref{lemm:monotone} holds before we deactivate $u$.
    Thus, $t$ is either an old leader, or $\height(t) \ge \height(u)$;
    and leaf $i$ lies weakly to the right of $v$.
    We deactivate $u$ and its right child. This replaces $u$ with three new leaders $u_1, u_2, u_3$
    with $\height(u_1) = \height(u)-1$ and $\height(u_2) = \height(u_3) = \height(u) - 2$.
    Thus, $\height(u_1) > \height(u_2) \ge \height(u_3) > \height(v)$.
    And, $t$ is either an old leader, or $\height(t) > \height(u_1)$.
    Moreover, among $t, u_1, u_2, u_3, v$, only $u_2$ and $u_3$ can be new leaders of the same height.
    This shows that \cref{lemm:monotone} still holds after the deactivations, so \cref{lemm:monotone} is maintained until the end.
\end{proof}

\begin{lemma} \label{lemm:access_correct}
    This implementation of \textsc{access}$(i)$ correctly maintains all invariants.
\end{lemma}
\begin{proof}
    We never deactivate leaf nodes, so invariant \ref{inv:1} holds.
    First, we deactivate all nodes on a root-to-leaf path. After that we only deactivate leader nodes,
    which by definition have a passive parent. Thus, invariant \ref{inv:2} holds.
    Whenever there are consecutive leaders $u, v$ with $|\height(u) - \height(v)| \ge 3$,
    we deactivate one of them. Thus, invariant \ref{inv:3} holds at the end of \textsc{access}$(i)$.
    Since we do not activate any nodes, invariant \ref{inv:5} is unaffected, and so is invariant \ref{inv:4} for all leaves $j \ne i$.
    Since these held before, they still hold afterwards.
    After \textsc{access}$(i)$, we have $s(i) = t-1$, so that $\lceil \log(1+t-s(i)) \rceil = 1$.
    Since we deactivated the parent of leaf $i$, invariant \ref{inv:4} also holds for leaf $i$.

    By \cref{lemm:monotone}, we create at most four leaders of every height.
    Thus, by invariant \ref{inv:c}, we increase the size of $L_i$ by at most $4 \cdot 4 \cdot 2^{i}$.
    Since $20 \cdot 4^{2^{i+2}} + 4 \cdot 4 \cdot 2^i \le 23 \cdot 4^{2^{i+2}}$, invariant \ref{inv:b} is preserved.
    We insert only leaders of height $\le 4 \cdot 2^i$ into the index tree $L_i$.
    Therefore, invariant \ref{inv:c} holds, and all newly created leaders are unpromoted.
    Let $u$ be a newly created leader, of height $(2 \cdot 2^i, 4 \cdot 2^i]$.
    By invariant \ref{inv:3}, there is a newly created leader $v$ of height $2 \cdot 2^i + 1$ or $2 \cdot 2^i+2$ between $u$ and leaf $i$.
    By \cref{lemm:contiguous,lemm:monotone}, all leaders between $u$ and $v$ have height $(2 \cdot 2^i, 4 \cdot 2^i]$.
    Thus, invariant \ref{inv:d} is maintained. This shows the lemma.
\end{proof}

\begin{lemma} \label{lemm:access_fast}
    The amortized running time of \textsc{access}$(i)$ is $O(d(i))$.
\end{lemma}
\begin{proof}
    We first analyze the time spent updating the main tree and the number of coins created.
    We generalize the notion of bills to larger height differences: A pair $u, v$ of consecutive leaders
    has $\max(0, |\height(u) - \height(v)| - 1)$ bills.
    By \cref{lemm:leader_unified}, the existing leader of leaf $i$ has height $O(d(i))$.
    Thus, deactivating the nodes on the root-to-leaf path of leaf $i$
    takes $O(d(i))$ time and creates $O(d(i))$ bills and coins.
    Consider now three consecutive leaders $t, u, v$ with $\height(u) - \height(v) \ge 3$.
    We deactivate $u$ and the right child of $u$. This creates three new leaders $u_1, u_2, u_3$
    with $\height(u_1) = \height(u)-1$ and $\height(u_2) = \height(u_3) = \height(u) - 2$.
    Overall, the number of bills decreases by at least one. Indeed, the pair $u_3, v$ has two fewer bills
    compared to the pair $u, v$, but the pair $t, u_1$ might have one more bill compared to $t, u$.
    We use this one bill to pay for the $O(1)$ coins of $u_1, u_2, u_3$ and for the $O(1)$ time spent.

    This shows that the time spent updating the main tree and the number of coins created have an amortized cost of $O(d(i))$.
    Next, we show that the time spent updating the index trees is $O(k)$ where $k$ is the number of newly created leaders.
    Then, this time is dominated by the time spent updating the main tree, and the lemma follows.
    Building the linked lists $N_i$ takes constant time per newly created leader.
    To analyze the time spent splitting or joining the index trees,
    let $h$ be the maximum height of any newly created leader.
    On each of the index trees $L_0, \dots, L_{O(\log h)}$, we perform a bulk deletion and a bulk insertion,
    involving $O(k)$ leaders in total.
    By \cref{lemm:leader_height}, each of these trees has height $O(h)$.
    Moreover, the heights of these trees are bounded by a geometric series.
    Thus, the bulk operations take $O(k + h)$ time in total.
    By invariant \ref{inv:3} and \cref{lemm:contiguous}, there are at least $h/2$ newly created leaders, i.e. $h \in O(k)$. This shows the lemma.
\end{proof}

To summarize, \cref{lemm:access_correct} shows that \textsc{access}$(i)$ is correct,
and \cref{lemm:access_fast} shows that \textsc{access}$(i)$ has the running time required by \cref{theo:tree_heap}.

\subsection{The \textsc{addLeaf} operation}

The \textsc{addLeaf} operation is easy to implement:
We activate leaf $n+1$ and insert it into the first index tree $L_0$.
Then, we increment $n$.
Finally, if $n = m$, then we double $m$ and increase the height of the main tree by $1$.
All newly created nodes are passive.

Activating leaf $n+1$ makes it a leader.
By invariant \ref{inv:5}, leaf $n$ is also a leader,
so invariant \ref{inv:3} is maintained.
We only add one leader to $L_0$, so invariant \ref{inv:b} is maintained.
It is easy to see that all other invariants are maintained.

As for the running time, leaf $n+1$ becomes a new leader, which creates one coin.
If we do not double $m$, then \textsc{addLeaf} takes $O(1)$ amortized time.
Indeed, by invariants \ref{inv:b} and \ref{inv:c}, the first index tree $L_0$ has constant height,
so inserting the new leader takes constant time.
Whenever we double $m$, we spend $O(m)$ time to create the new nodes.
By the standard amortization argument for dynamic arrays~\cite{cormen01introduction}, this amounts to amortized $O(1)$ time per \textsc{addLeaf} operation.
This shows:

\begin{theorem} \label{theo:addleaf}
    This implementation of 
    \textsc{addLeaf} correctly maintains all invariants
    and has an amortized running time of $O(1)$.
\end{theorem}

\subsection{The \textsc{cleanup} operation}

We first describe a subroutine \textsc{cleanupStep}$(i)$.
This subroutine reduces the size of $L_i$ back
to $\le 16 \cdot 4^{2^{i+2}}$, in amortized zero time,
but it might increase the size of $L_{i+1}$ in the process.
We achieve this by both moving unpromoted blocks from $L_i$ to $L_{i+1}$,
and by activating as many nodes in the main tree as possible, without violating invariants \ref{inv:4} or \ref{inv:c}.

To perform \textsc{cleanupStep}$(i)$, first build a linked list $A$ of all leaders in $L_0, \dots, L_i$ in left-to-right order.
For every leader of height $4 \cdot 2^i$ in $A$, find its unpromoted block of leaders of height $>2 \cdot 2^i$ in $A$,
remove the block from $A$ and insert it into $L_{i+1}$.
Then, while there are consecutive leaders $t, u, v, w$ where $u, v$ have the same parent $p$,
activate $p$ unless (a) $\height(p) > \min(\height(t), \height(w)) + 1$; (b) this would violate invariant \ref{inv:4} or \ref{inv:5};
or (c) $\height(p) > 4 \cdot 2^i$.
If we do activate $p$, replace $u, v$ by $p$ in the linked list.
Finally, we rebuild the index trees:
For every $0 \le j \le i$, build a balanced binary tree $L_j$
that contains all (remaining) leaders in $A$ of height $(2 \cdot 2^j, 4 \cdot 2^j]$,
or height $[0, 4]$ for $j=0$. This concludes \textsc{cleanupStep}$(i)$.

To efficiently find a parent $p$ to activate during \textsc{cleanupStep}$(i)$,
we maintain a pointer to the leader $t$
and perform a left-to-right scan through $A$. Whenever we fail to activate $p$, we move this pointer one step to the right.
Whenever we successfully activate $p$, we move this pointer two steps to the left.
When the pointer reaches the end of $A$, there are no more parents $p$ to activate.

We now describe the full operation \textsc{cleanup}.
Recall invariant \ref{inv:b}.
The goal of \textsc{cleanup} is to reduce the size of all index trees from up to $23 \cdot 4^{2^{i+2}}$ to at most $20 \cdot 4^{2^{i+2}}$.
To perform \textsc{cleanup}, find the maximum index $i$ with $|L_i| > 20 \cdot 4^{2^{i+2}}$,
and call \textsc{cleanupStep}$(i)$.
Repeat this until no such index $i$ exists.
A naive implementation of \textsc{cleanup} spends $O(\log \log n)$ time per iteration looking for the index $i$.
We avoid this overhead by keeping all index trees $L_i$ with $|L_i| > 20 \cdot 4^{2^{i+2}}$
in a linked list and updating this linked list whenever we update an index tree during
\textsc{addLeaf}, \textsc{access}, or \textsc{cleanupStep}.
With this bookkeeping, the running time of \textsc{cleanup} is $O(1)$ plus the time spent calling \textsc{cleanupStep}.

\begin{lemma} \label{lemm:post_size}
    After \textsc{cleanupStep}$(i)$, for every $k < 4 \cdot 2^i$,
    there are at most $16 \cdot 4^k$ leaders of height $\le k$.
    And, there are at most $12 \cdot 4^{2^{i+2}}$ leaders of height exactly $4 \cdot 2^i$.
    In particular, $|L_j| \le 16 \cdot 4^{2^{j+2}}$ for every $0 \le j \le i$.
\end{lemma}
\begin{proof}
    We say a leader is \emph{locked} if its parent cannot be activated due to invariant \ref{inv:4} or~\ref{inv:5}.
    We first analyze the number of locked nodes: There are at most two locked nodes of height $\le k$
    for every leaf $i$ with $\lceil \log_2(1+t-s(i)) \rceil \le k+1$,
    or equivalently, $t-s(i) \le 2^{k+1} - 1$. Since the $s(i)$ values are pairwise distinct
    and satisfy $s(i) < t$, there are at most $2^{k+1}-1$ such leaves.
    Leaves with a smaller $s(i)$ value
    can only create locked nodes of height $>k$.
    Invariant \ref{inv:5} creates one or two additional locked nodes.
    Thus, for every integer $k$, there are at most $2 \cdot (2^{k+1}-1) + 2 \le 2^{k+2}$ locked nodes of height $\le k$.

    Consider the state of the main tree after running \textsc{cleanupStep}$(i)$.
    Let $s, t, u, v, w$ be consecutive leaders and let $p$ be the parent of $u$.
    Suppose that $u$ is not locked and that $\height(u) < 4 \cdot 2^i$.
    If moreover $\min(\height(s), \height(t), \height(v), \height(w)) \ge \height(u)$,
    then $p$ would get activated during \textsc{cleanupStep}$(i)$. But $p$ was not activated,
    thus, one of the leaders $s, t, v, w$ has height at most $\height(u) - 1$.
    To summarize: for a leader $u$ of height $< 4 \cdot 2^i$ that is not locked,
    there is a leader of lower height within at most two steps to the left or right.
    By applying this argument inductively, we find a locked leader within at most $2 \height(u)$ steps from $u$.
    Thus, for every $k < 4 \cdot 2^i$,
    all leaders of height $\le k$ are within $2 k$ steps from a locked node of height $\le k$.
    As there are $2^{k+2}$ such locked nodes,
    the total number of leaders of height $\le k$ is bounded by $(4k+1) \cdot 2^{k+2} \le 16 \cdot 4^{k}$.

    Let $1 \le j < i$ be arbitrary. By \ref{inv:c}, the index tree $L_j$ contains only leaders of height $\le 2^{j+2}$.
    Thus, by the first part of the proof, $|L_j| \le 16 \cdot 4^{2^{j+2}}$.
    To analyze the size of $L_i$, we also need to bound the number of leaders of height exactly $4 \cdot 2^i$.
    In the first step of \textsc{cleanupStep}$(i)$, we promote all leaders of height $4 \cdot 2^i$
    to $L_{i+1}$. Afterwards, we can only create a leader of height $4 \cdot 2^i$
    by destroying two leaders of height $4 \cdot 2^i - 1$.
    By \ref{inv:b}, we started with $\sum_{j=0}^{i} 23 \cdot 4^{2^{j+2}} \le 24 \cdot 4^{2^{i+2}}$ leaders,
    so we created at most
    $12 \cdot 4^{2^{i+2}}$ leaders of height exactly $4 \cdot 2^i$.
    As there are at most $16 \cdot 4^{2^{i+2} - 1} = 4 \cdot 4^{2^{i+2}}$ leaders of height $\le 4 \cdot 2^{i}-1$,
    this shows $|L_i| \le 16 \cdot 4^{2^{i+2}}$.
\end{proof}

\begin{lemma} \label{lemm:cleanup_correct}
    Let $i$ be the maximum index with $|L_i| \ge 20 \cdot 4^{2^{i+2}}$.
    After \textsc{cleanupStep}$(i)$, all invariants hold, with the following refinement for \ref{inv:b}:
    We have $|L_j| \le 20 \cdot 4^{2^{j+2}}$ for all $j \ne i+1$,
    and $|L_{i+1}| \le 23 \cdot 4^{2^{(i+1)+2}}$.
\end{lemma}
\begin{proof}
    Invariant \ref{inv:1} holds since we never activate or deactivate leaves,
    and Invariant \ref{inv:2} holds since we only ever activate the parent of two leaders.
    If $t, u, v, w$ are consecutive leaders where $u, v$ have the same parent $p$,
    then we only activate $p$ if $\height(p) \le \min(\height(t), \height(w)) + 1$.
    Thus, invariant \ref{inv:3} holds.
    We explicitly check for invariants~\ref{inv:4} and~\ref{inv:5},
    so these hold too.
    \cref{lemm:post_size} shows that invariant \ref{inv:b} holds for $j \le i$.
    For $j \ge i+2$, we do not change $L_j$ at all, so $|L_j| \le 20 \cdot 4^{2^{j+2}}$ by maximality of $i$.
    We only change $L_{i+1}$, by promoting leaders from $L_i$.
    Thus, $|L_{i+1}| \le 23 \cdot 4^{2^{i+2}} + 20 \cdot 4^{2^{(i+1)+2}} \le 21 \cdot 4^{2^{(i+1)+2}}$ by maximality of $i$. This shows the refinement of invariant \ref{inv:b}.
    When rebuilding the index trees, we place the leaders of height $(2 \cdot 2^j, 4 \cdot 2^j]$
    into $L_j$. Thus, invariants \ref{inv:c} and~\ref{inv:d} hold.
\end{proof}
\begin{lemma} \label{lemm:cleanup_time}
    The amortized running time of \textsc{cleanupStep}$(i)$ is zero.
\end{lemma}
\begin{proof}
    We first show that \textsc{cleanupStep}$(i)$ does not create any bills:
    Moving leaders to $L_{i+1}$ has no effect on the number of bills.
    For consecutive leaders $t, u, v, w$ where $u, v$ have the same parent $p$,
    we only activate $p$ if $\height(p) \le \min(\height(t), \height(w)) + 1$.
    Thus, this activation does not create any bills either.

    Next, we analyze the number of coins.
    Let $N$ be the total size of $L_0, \dots, L_i$ at the start of \textsc{cleanupStep}$(i)$,
    and let $N'$ be the total size of $L_0, \dots, L_i$ at the end of \textsc{cleanupStep}$(i)$.
    We claim that $N-N' \in \Omega(N)$.
    We have $N \ge 20 \cdot 4^{2^{i+2}}$ since we only run \textsc{cleanupStep}$(i)$ if $|L_i| \ge 20 \cdot 4^{2^{i+2}}$.
    And, by \cref{lemm:post_size},  we have $N' \le \sum_{j=0}^{i} 16 \cdot 4^{2^{j+2}} \le 17 \cdot 4^{2^{i+2}}$.
    Thus, $N' \le \frac{17}{20} N$, so $N - N' \ge \frac{3}{20} N$.

    For every leader $u$ that we move to $L_{i+1}$, we get one coin since $u$ stops being unpromoted.
    Also, whenever we activate a node $p$ with children $u, v$, we get at least one coin.
    Indeed, we get four coins for destroying the leaders $u, v$, but
    we need two or three coins for the new leader $p$, depending on whether $p$ is unpromoted.
    These are the only two operations that decrease the total size of $L_0, \dots, L_i$.
    Thus, we get at least $N - N' \ge \frac{3}{20} N$ coins in total.
    In particular, during \textsc{cleanupStep}$(i)$, we may spend $O(N)$ time for free.

    When promoting leaders from $A$ to $L_{i+1}$,
    we always promote a block of $\Omega(2^i)$ leaders.
    Bulk inserting these leaders takes $O(\log |L_{i+1}|)$ time, plus constant time per leader.
    By \ref{inv:b}, $\log |L_{i+1}| \in O(2^i)$.
    Thus, we spend constant time per promoted leader, so $O(N)$ time in total.

    Afterwards, we activate some nodes, and rebuild the trees $L_0, \dots, L_i$.
    Since every node in the main tree stores its height and the maximum access time of a leaf in its subtree,
    we can check in $O(1)$ time whether a given parent $p$ can be activated.
    Whenever we do activate $p$, the size of the linked list shrinks by one, so this happens $O(N)$ times.
    We also spend $O(N)$ time scanning through the linked list, moving the pointer to $t$, and unsuccessfully activating nodes.
    Building a balanced binary tree takes linear time in its size.
    Thus, rebuilding $L_0, \dots, L_i$ takes $O(N)$ time.
\end{proof}

\begin{theorem} \label{theo:cleanup}
    \textsc{Cleanup} maintains all invariants, has an amortized running time of $O(1)$,
    and restores the stronger bound in invariant~\ref{inv:b}.
\end{theorem}
\begin{proof}
    \textsc{Cleanup} only modifies the data structure via \textsc{cleanupStep}.
    By \cref{lemm:cleanup_correct}, \textsc{cleanupStep} maintains all invariants.
    And \textsc{cleanup} only terminates once the stronger bound in invariant~\ref{inv:b} holds for all index trees. Thus, all invariants are maintained.

    \textsc{Cleanup} spends $O(1)$ time checking whether there are any index trees that are too big,
    in which case it performs a \textsc{cleanupStep}. This is repeated until no \textsc{cleanupStep} happens.
    If \textsc{cleanup} performs at least one \textsc{cleanupStep},
    then the $O(1)$ overhead of \textsc{cleanup} is dominated by the time spent in \textsc{cleanupStep},
    so the amortized running time of \textsc{cleanup} is zero by \cref{lemm:cleanup_time}.
    If no \textsc{cleanupStep} happens, then the running time of \textsc{cleanup} is $O(1)$.
\end{proof}

\subsection{Putting things together}

After every \textsc{addLeaf} or \textsc{access} operation,
we run \textsc{cleanup}.
We show that our construction achieves \cref{theo:tree_heap}:

\thmTree*

\begin{proof}
    By \cref{lemm:leader_unified,lemm:leader_height} every leaf $i$ has depth $O(d(i))$.
    By \cref{theo:addleaf}, \textsc{addLeaf} runs in amortized $O(1)$ time.
    By \cref{lemm:access_fast}, \textsc{access} runs in amortized $O(d(i))$ time.
    By \cref{theo:cleanup}, running \textsc{cleanup} after these operations does not increase the amortized running time.
    Since the values of items have no effect on the structure of our heap,
    we only need to update the subtree minima during \textsc{changeKey}.
    This takes worst case $O(d(i))$ time.
    Since we run \textsc{cleanup} after every \textsc{addLeaf} or \textsc{access} operation,
    \cref{lemm:access_correct,theo:addleaf,theo:cleanup} show that we maintain all invariants.
\end{proof}

\section{A Slightly Stronger Unified Bound}

There are multiple slightly different ways of defining the working set property of heaps \cite{iacono_improved_2000,haeupler2024fast,elmasry_priority_2006}.
The simplest one is that $x \gets\textsc{extractMin}()$ takes $O(\log k)$ time
where $k$ is the number \textsc{insert} operations while $x$ is in the heap.
In the stronger, but more complicated definition, of Elmasry~\cite{elmasry_priority_2006}, the value $k$ only counts
\textsc{insert} operations for which the inserted item remains in the heap until $x$ is extracted.
In other words, the value $k$ skips items that are extracted before $x$ is extracted.

We designed \cref{def:unified} to be a simple version of the unified bound for Heaps
which is sufficient for our result on sorting under partial information.
However, our heap construction can be adapted to achieve a stronger but more complicated version of the unified bound.
Recall that we consider a sequence of \textsc{insert} and \textsc{extractMin} operations
where an item $x$ is the $a(x)$-th item to be inserted and the $b(x)$-th item to be extracted.
We modify \cref{def:unified} and decrease the term $|a(x) - a(y)|$ by only counting items
inserted between $x$ and $y$ that are still in the heap. This yields:

\begin{definition} \label{def:unified_immediate}
    A heap data structure has the \emph{immediate unified bound}
    if it supports \textsc{insert}$(x)$ in amortized $O(1)$ time,
    the very first \textsc{extractMin} operation in amortized $O(1)$ time,
    and subsequent $x \gets$\textsc{extractMin}$()$ operations in amortized time
    $$
        O\Big(\min\limits_{\substack{y \in X\\b(y) < b(x)}} \big(1 + \log|A(x, y)| + \log (b(x) - b(y))\big)\Big),
    $$
    where $A(x, y) := \{z \in X \mid b(z) \ge b(x) \land \min(a(x), a(y)) \le a(z) \le \max(a(x), a(y))\}$.
\end{definition}

We sketch how the construction in \cref{sect:heap_construction} can be adapted to the immediate unified bound.
This time, we consider up to $n$ leaves, numbered by some subsequence of $[n]$.
Non-existent numbers correspond to leaves that have been deleted.
Let $t$ be a time counter that starts at $0$.
Every integer $i \in [n]$ has a last access time $s(i) < t$.
We define $\ol{d}(i) := 1 + \min_{1 \le j \le n} (\log (1 + |\{k \mid \text{leaf $k$ exists} \land \min(i, j) \le k \le \max(i, j)\}| + \log (1 + t - s(j)))$.
Note that every integer $j$ can be used as a finger, even if no leaf $j$ exists.

\begin{restatable}{theorem}{thmTreeIm}[Improved \cref{theo:tree_heap}] \label{theo:tree_heap_immediate}
    There is a data structure that maintains a rooted binary tree
    on top of up to $n$ leaves numbered by some subsequence of $[n]$ that each store an item,
    such that each node in the tree stores its subtree minimum
    and every leaf $i$ has depth $O(\ol{d}(i))$.
    The data structure supports the operations,
    where $i$ is an integer for which a leaf exists:
    \begin{itemize}
        \item \textsc{addLeaf}$(x)$: In amortized $O(1)$ time, increment $n$ and add a new leaf with index $n$, value $x$, and $s(n) = -\infty$.
        \item \textsc{access}$(i)$: In amortized $O(\ol{d}(i))$ time (computed before the operation), set $s(i) = t$ and then increment $t$. For the very first access operation, $\ol{d}(i)$ is ill-defined, so we require a bound of amortized $O(1)$ time instead.
        \item \textsc{changeKey}$(i, x)$: In $O(\ol{d}(i))$ time, replace the item stored at leaf $i$ with $x$.
        \item \textsc{deleteLeaf}$(i)$: In $O(\ol{d}(i))$ time, delete the leaf $i$.
    \end{itemize}
    The data structure may remodel the internal tree nodes arbitrarily (within the allowed time bounds).
    The leaves do \emph{not} need to be stored in order.
\end{restatable}

Given \cref{theo:tree_heap_immediate}, we obtain a heap with the immediate unified bound (\cref{def:unified_immediate})
as before, except that \textsc{extractMin} invokes $\textsc{deleteLeaf}(i)$ instead of $\textsc{changeKey}(i, \infty)$.

\subsection{Modifying our Heap Construction}

We briefly sketch the changes necessary to go from \cref{theo:tree_heap} to \cref{theo:tree_heap_immediate}.
We call an integer $i \in [n]$ \emph{alive} if the leaf $i$ still exists, and \emph{dead} otherwise.
Each leaf of the main tree corresponds to an integer in $[n]$ that is either dead or alive,
but among two consecutive leaves, we require at least one of them to be alive.
In fact, for every maximal block of consecutive dead integers, we only store the most recently accessed integer as a leaf in the main tree.
We also store all alive integers.
This ensures that, for any two integers $i, j$, the number of leaves in the main tree with integers between $i$ and $j$ is at most $2 \cdot |\{k \mid \text{leaf $k$ exists} \land \min(i, j) \le k \le \max(i, j)\}|$,
which then yields the stronger bound $\ol{d}(i)$ instead of $d(i)$.

For the main tree, we use a 2-3-4 tree~\cite{huddleston1982new} instead of a complete binary tree.
This will allow us to delete leaves corresponding to consecutive dead integers created by the $\textsc{deleteLeaf}$ operation.
The index trees remain the same, but the constants in invariant \ref{inv:b} need to be increased since 2-3-4 trees allow nodes to have up to four children.
Similarly, some constants in the analysis need to be increased,
e.g. in \cref{lemm:monotone}, we may get four consecutive newly created leaders of the same height,
and in the proof of \cref{lemm:post_size}, we might need $6 \height(u)$ steps from $u$ to find a locked leader.

The implementations of \textsc{addLeaf}, \textsc{access}, \textsc{changeKey}, \textsc{cleanupStep} and \textsc{cleanup} remain largely the same.
To perform \textsc{deleteLeaf}$(i)$, we first invoke $\textsc{changeKey}(i, \infty)$.
Then, we mark leaf $i$ as dead and delete any adjacent dead leaves in the main tree.
This might cause the 2-3-4 tree to perform some rebalancing operations, specifically \emph{fusing} and \emph{borrowing}~\cite{huddleston1982new}.
Rebalancing non-leader nodes is a non-issue.
For any rebalancing operation involving one or more leader nodes, we first deactivate these leader nodes and perform the operation.
Then, we fix invariant \ref{inv:3} as in the \textsc{access} operation, and update the index trees.
If $h$ is the height of the node involved in the rebalancing operation, then deactivating leaders creates $O(1)$ bills, fixing invariant \ref{inv:3} takes zero amortized time, and updating the index trees takes $O(h)$ time.
Thus, the amortized running time of a rebalancing operation at height $h$ is $O(h)$.
The 2-3-4 tree analysis (Theorem 2.b in \cite{huddleston1982new}) shows that $n$ insertions and deletions
trigger at most $n \cdot (\frac{3}{5})^{h-1}$ rebalancing operations at height $h$.
Thus, the total additional running time due to rebalancing is $n\cdot \sum_{h=1}^{\log_2 n} h (\frac{3}{5})^{h-1} \in O(n)$.

\section{Open problems}

\subsection{Heaps with the Unified Bound}

Both the construction and analysis of our heap with the unified bound is complicated
and uses a slew of invariants and amortization potentials.
It would be nice to have a simple construction, as was achieved for the working set heap \cite{van_der_hoog_simpler_2025-1}.
Alternatively, it would be interesting to show that an existing self-adjusting structure, such as the smooth heap or pairing heap does the trick.

\subsection{Chain-Antichain Partition for Points}

An earlier attempt at solving sorting under partial information with merge sort (instead of heap sort)
failed since we could not solve the following subproblem:
You are given an integer $k > 0$ and a list $X$ of $n$ items. In $O(n \log k)$ time, partition $X$ into

\begin{itemize}
    \item up to $k$ increasing subsequences of any length,
    \item plus any number of decreasing subsequences, each of length at least $k+1$.
\end{itemize}

Such a partition always exists: While the longest decreasing subsequence of $X$ is longer than $k$,
delete it.
Afterwards, by Dilworth's theorem, $X$ can be partitioned into $k$ increasing subsequences.

For our purpose, it would be enough if the decreasing subsequences had length $\Omega(k^{\varepsilon})$ for some $\varepsilon > 0$.
The difficult part is that there is no wiggle room in terms of running time: it has to be $O(n \log k)$.
(The chain-antichain merge sort of \cite{HoogRutschmann24} solves a generalization of this problem to arbitrary partial orders in $\tilde{O}(n k^2)$ time.)

\bibliographystyle{plain}
\bibliography{refs,zotero}
\end{document}